\documentclass[letterpaper,11pt]{article}

\author{\Large\textbf{Artur Czumaj} \hspace{4mm} \textbf{Peter Davies} \\[0.10in]
Department of Computer Science \\
Centre for Discrete Mathematics and its Applications 
\\
University of Warwick}

\title{\textbf{Exploiting Spontaneous Transmissions for Broadcasting and Leader Election in Radio Networks}
\thanks{Research partially supported by the Centre for Discrete Mathematics and its Applications (DIMAP).}
\thanks{Contact information: \{A.Czumaj, P.W.Davies\}@warwick.ac.uk. Phone: +44 24 7657 3796.}
\\[0.25in]
}

\date{}

\usepackage{amsmath,amssymb,amsfonts,amsthm}
\usepackage{xspace}
\usepackage{algorithm}
\usepackage{algorithmicx}
\usepackage{algpseudocode}
\usepackage{thmtools}
\usepackage{cleveref}
\usepackage{times}\usepackage[scaled=0.92]{helvet} 
\usepackage[T1]{fontenc}
\usepackage{bm}
\usepackage[margin=1.0in]{geometry}

\newtheorem{theorem}{Theorem}

\newtheorem{lemma}[theorem]{Lemma}

\newtheorem{claim}[theorem]{Claim}

\newcommand{\junk}[1]{{}}
\newcommand{\prob}[1]{\mathbb{P}[#1]}
\newcommand{\Prob}[1]{\mathbb{P}\left[#1\right]}
\newcommand{\expect}[1]{\mathbb{E}[#1]}

\newcommand{\net}{\ensuremath{\mathfrak{N}}\xspace}

\newcommand{\xx}{\ensuremath{\bm{x}}}

\begin{document}
\begin{titlepage}
\clearpage\maketitle\thispagestyle{empty}

\begin{abstract}
We study two fundamental communication primitives: \emph{broadcasting} and \emph{leader election} in the classical model of multi-hop radio networks with unknown topology and without collision detection mechanisms.
It has been known for almost 20 years that in undirected networks with $n$ nodes and diameter $D$, randomized broadcasting requires $\Omega(D \log\tfrac{n}{D} + \log^2n)$ rounds in expectation, assuming that uninformed nodes are not allowed to communicate (until they are informed). Only very recently, Haeupler and Wajc (PODC'2016) showed that this bound can be slightly improved for the model with spontaneous transmissions, providing an $O(D\frac{\log n\log\log n}{\log D} + \log^{O(1)}n)$-time broadcasting algorithm. In this paper, we give a new and faster algorithm that completes broadcasting in $O(D\frac{\log n}{\log D} + \log^{O(1)}n)$ time, with high probability. This yields the first optimal $O(D)$-time broadcasting algorithm whenever $D$ is polynomial in $n$.

Furthermore, our approach can be applied to design a new leader election algorithm that matches the performance of our broadcasting algorithm. Previously, all fast randomized leader election algorithms have been using broadcasting as their subroutine and their complexity have been asymptotically strictly bigger than the complexity of broadcasting. In particular, the fastest previously known randomized leader election algorithm of Ghaffari and Haeupler (SODA'2013) requires $O(D \log\frac{n}{D} \min\{\log\log n, \log\frac nD\} +\log^{O(1)}n)$-time with high probability. Our new algorithm requires $O(D\frac{\log n}{\log D} + \log^{O(1)}n)$ time with high probability, and it achieves the optimal $O(D)$ time whenever $D$ is polynomial in $n$.
\end{abstract}

\bigskip
\bigskip
\bigskip
\bigskip

\end{titlepage}


\section{Introduction}


\subsection{Model of communication networks}
\label{subsec:model}

We consider the classical model of \emph{ad-hoc radio networks} with \emph{unknown structure}. A \emph{radio network} is modeled by an \emph{undirected} network $\net = (V,E)$, where the set of nodes corresponds to the set of transmitter-receiver stations. An edge $\{v,u\} \in E$ means that node $v$ can send a message directly to node $u$ and vice versa. To make propagation of information feasible, we assume that
\net is connected.

In accordance with the standard model of unknown (ad-hoc) radio networks (for more elaborate discussion about the model, see, e.g., \cite{-ABLP91,-BGI92,-CGGPR00,-CGR00,-CMS03,-GH13,-GHK13,-HW16,-KP03,-KM98,-Pel07}), we make the assumption that a node does not have any  prior knowledge about the topology of the network, its in-degree and out-degree, or the set of its neighbors. We assume that the only knowledge of each node is the \emph{size} of the network $n$ and the \emph{diameter} of the network $D$.

Nodes operate in discrete, synchronous time steps. When we refer to the ``running time'' of an  algorithm, we mean the number of time steps which elapse before completion (i.e., we are not concerned with the number of calculations nodes perform within time steps). In each time step a node can either \emph{transmit} a message to all of its out-neighbors at once or can remain silent and \emph{listen} to the messages from its in-neighbors. We do not make any restriction on the size of messages, though the algorithms we present can easily be made to operate under the condition of $O(\log n)$-bit transmissions.

A further important feature of the model considered in this paper is that it allows \emph{spontaneous transmissions}, that is, any node can transmit if it so wishes. In some prior works (see, e.g., \cite{-CGGPR00,-CR06,-KM98}), it has been assumed (typically for the broadcasting problem) that uninformed nodes are not allowed to communicate (until they are informed). While this assumption can be of interest for the broadcasting problem, it is meaningless for the leader election problem, and so, throughout this paper we will allow spontaneous transmissions.

The distinguishing feature of radio networks is the interfering behavior of transmissions. In the most standard radio networks model, the \emph{model without collision detection}  (see, e.g., \cite{-ABLP91,-BGI92,-CMS03,-Pel07}), which is studied in this paper, if a node $v$ listens in a given round and precisely one of its in-neighbors transmits, then $v$ receives the message. In all other cases $v$ receives nothing; in particular, the lack of collision detection means that $v$ is unable to distinguish between zero of its in-neighbors transmitting and more than one.

The model without collision detection describes the most restrictive interfering behavior of transmissions; also considered in the literature is a less restrictive variant, the model with collision detection, where a node listening in a given round can distinguish between zero of its in-neighbors transmitting and more than one (see, e.g., \cite{-GHK13,-Pel07}).


\subsection{Key communications primitives: Broadcasting, leader election, and \textsc{Compete}}

In this paper we consider two fundamental communications primitives, namely \emph{broadcasting} and \emph{leader election}. We present randomized algorithms that perform these tasks \emph{with high probability} (i.e., $1-n^{-c}$ for an arbitrary large constant $c$), and analyze worst-case running time.

\emph{Broadcasting} is one of the most fundamental problems in communication networks and has been extensively studied for many decades (see, e.g., \cite{-Pel07} and the references therein).
The premise of the broadcasting task is that one particular node, called the \emph{source}, has a message which must become known to all other nodes. As such, broadcasting is one of the most basic means of global communication in a network.

\emph{Leader Election} is another most fundamental problems in communication networks that aims to ensure that all nodes agree on such a designated leader. Specifically, at the conclusion of a leader election algorithm, all nodes should output the same node ID, and precisely one node should identify this ID as its own. Leader election is a fundamental primitive in distributed computations and, as the most basic means of breaking symmetry within radio networks, it is used as a preliminary step in many more complex communication tasks. For example, many fast multi-message communication protocols require construction of a breadth-first search tree (or some similar variant), which in turn requires a single node to act as source (for more examples, cf. \cite{-CKP12,-GH13}, and the references therein).

To design efficient algorithms for broadcasting and leader election, we will be studying an auxiliary problem that we call \textsc{Compete}. \textsc{Compete} has a similar flavor to broadcasting, but instead of transmitting a single message from a single source to all nodes in the network, it takes as its input a source set $S \subseteq V$, in which every source $s \in S$ has a message (of integer value) it wishes to propagate, and guarantees that upon completion all nodes in \net know the highest-valued source message.

It is easy to see how the \textsc{Compete} process generalizes broadcasting: it is simply invoked with the source as the only member of the set $S$. To perform leader election, one can probabilistically generate a small set (e.g., of size $\Omega(\log n)$) of candidate leaders, and then perform \textsc{Compete} using this set, with $ID$s as the messages to be propagated. Therefore, to design an efficient randomized broadcasting and leader election algorithms, it is sufficient to design a fast randomized algorithm for \textsc{Compete} (cf. Section \ref{sec:broadcasting-and-leader-election}).


\subsection{Previous work}

As a fundamental communications primitive, the task of \emph{broadcasting} has been extensively studied for various network models, see, e.g., \cite{-Pel07} and the references therein.

For the model studied in this paper, undirected radio networks with unknown structure and without collision detection, the first non-trivial major result was due to Bar-Yehuda et al.\ \cite{-BGI92}, who, in a seminal paper, designed an almost optimal randomized broadcasting algorithm achieving the running time of $O((D + \log n) \cdot \log n)$ with high probability. This bound was later improved by Czumaj and Rytter \cite{-CR06}, and independently Kowalski and Pelc \cite{-KP03b}, who gave randomized broadcasting algorithms that complete the task in $O(D \log \frac{n}{D} + \log^2 n)$ time with high probability. Importantly, all these algorithms were assuming that nodes are not allowed to transmit spontaneously, i.e., they must wait to receive the source message before they can begin to participate. Indeed, for the model with no spontaneous transmissions allowed, it has been known that any randomized broadcasting algorithm requires $\Omega(D \log \frac{n}{D} + \log^2 n)$ time \cite{-ABLP91,-KM98}.
Only very recently, Haeupler and Wajc \cite{-HW16} demonstrated that allowing spontaneous transmissions can lead to faster broadcasting algorithms, by designing a randomized algorithm that completes broadcasting in $O(D \frac{\log n \log \log n}{\log D} + \log^{O(1)}n)$ time, with high probability. This is the only algorithm (that we are aware of) that beats the lower bound of $\Omega(D \log \frac{n}{D} + \log^2 n)$ \cite{-ABLP91,-KM98} in the model with no spontaneous transmissions. Given that for the model that allows spontaneous transmissions any broadcasting algorithm requires $\Omega(D + \log^2)$ time (cf. \cite{-ABLP91,-Pel07}), the algorithm due to Haeupler and Wajc \cite{-HW16} is \emph{almost optimal} (up to an $O(\log\log n)$ factor) whenever $D$ is polynomial in $n$.

Broadcasting has been also studied in various related models, including directed networks, deterministic broadcasting protocols, models with collision detection, and models in which the entire network structure is known. For example, in the model with collision detection, an $O(D+\log^6 n)$-time randomized algorithm due to Ghaffari et al. \cite{-GHK13} is the first to exploit collisions and surpass the algorithms 
for broadcasting without collision detection. For deterministic protocols, the best results are an $O(n\log D \log\log D)$-time algorithm in directed networks \cite{-CD16a}, and an $O(n\log D)$-time algorithm in undirected networks \cite{-K05}.

For more details about broadcasting in various model, see, e.g., \cite{-Pel07} and the references therein.

The problem of \emph{leader election} has also been extensively studied in the distributed computing community for several decades. For the model considered in this paper, it is known that a simple reduction (see, e.g., \cite{-BGI91}), involving performing a network-wide binary search for the highest ID using broadcasting as a subroutine every step, requires $O(T_{BC}\log n)$ time. Here $T_{BC}$ is time taken to perform broadcasting (provided the broadcasting algorithm used can be extended to work from multiple sources). This yields leader election randomized algorithms taking time $O(D\log\frac{n}{D}\log n+\log^3 n)$ using the broadcasting algorithms of \cite{-CR06,-KP03b}, or $O(D\frac{\log^2 n \log\log n}{\log D}+\log^{O(1)} n)$ using the broadcasting algorithm of \cite{-HW16}. This approach has been improved only very recently by Ghaffari and Haeupler \cite{-GH13}, who took a more complex approach to achieve an $O(D\log\frac{n}{D}+\log^3 n) \cdot \min\{\log\log n,\log\frac nD\}$ time algorithm based on growing clusters within the network.
%
%
Notice that in the regime of large $D$ being polynomial in $n$, when $D \approx n^c$ for a constant $c$, $0 < c < 1$, the fastest leader election algorithm achieves the (high probability) running time of $O(D \log n \log\log n)$.

Leader election has also been studied in various related settings. For example, one can achieve $O(T_{BC})$ expected (rather than worst case) running time \cite{-CD16b}, or time $O(T_{BC} \sqrt{\log n})$ with high probability even for directed networks \cite{-CD16b}, and deterministically time $O(n\log n \log D \log\log D )$ \cite{-CD16a} or $O(n \log^{3/2}n \sqrt{\log\log n})$~\cite{-CKP12}.

\junk{
Leader election has also been studied in various related settings. For example, there exists a protocol which achieves $O(T_{BC})$ expected (rather than worst case) running time \cite{-CD16b}, which achieves time $O(T_{BC} \sqrt{\log n})$ with high probability even for directed networks \cite{-CD16b}, and deterministic protocols taking time $O(n\log n \log D \log\log D )$ (using the reduction method with the deterministic broadcasting algorithm of \cite{-CD16a}) or $O(n \log^{3/2}n \sqrt{\log\log n})$ \cite{-CKP12}.
}


\subsection{New results}

In this paper we extend the approach recently developed by Haeupler and Wajc \cite{-HW16} to design a fast randomized algorithm for \textsc{Compete}, running in time $O(\frac{D\log n}{\log D}+|S|D^{0.125} +\log^{O(1)}n)$, with high probability (Theorem \ref{thm:compete}). By applying this algorithm to the broadcasting problem (Theorem \ref{thm:broadcasting}) and to the leader election problem (Theorem \ref{thm:leader-election}), we obtain randomized algorithms for both these problems running in time $O(D \frac{\log n}{\log D} + \log^{O(1)}n)$, with high probability. For $D = \Omega(\log^cn)$ for a sufficiently large constant $c$, these running time bounds improve the fastest previous algorithms for broadcasting and leader election by factors $O(\log\log n)$ and $O(\log n \log\log n)$, respectively. More importantly, whenever $D$ is polynomial in $n$ (i.e., $D = \Omega(n^c)$, for some positive constant $c$), this running time is $O(D)$, which is optimal since $D$ time is required for any information to traverse the network.


Our algorithms are the first to achieve optimality over this range of parameters, and are also the first instance (in our model) of leader election time being equal to fastest broadcasting time, since the former is usually a harder task in radio network models.

Finally, even though the current lower bounds for the randomized broadcasting and leader election problems are $\Omega(D + \log^2n)$, we would not be surprised if our upper bounds $O(D \frac{\log n}{\log D} + \log^{O(1)}n)$ were tight for $D = \Omega(\log^c n)$ for some sufficiently large constant $c$.

\emph{Note:} We assume throughout that $D = \Omega(\log^c n)$ for some sufficiently large constant $c$. If this is not the case, then the $O(D\log\frac nD + \log^2 n )$-time algorithm of \cite{-CR06,-KP03b} should be used instead.


\section{Approach}

Our approach to study \textsc{Compete} (and hence also broadcasting and leader election problems) follows the methodology recently applied for fast distributed communication primitives by Ghaffari, Haeupler, and others (see, e.g., \cite{-GH13,-HW16}). In order to solve the problem, we split computations into three parts. First, all nodes in the network will try to build some basic information about local structure, and will locally create some clustering of the network. Then, using this clustering, the nodes will perform some computations within each cluster, so that all nodes in the cluster share some useful knowledge. Finally, the knowledge from the clusters will be utilized to efficiently perform global communication.


\subsection{Clusterings, \textsc{Partition}, and schedulings}

To implement this approach efficiently, we follow a similar line to that of Haeupler and Wajc \cite{-HW16} and rely on a clustering procedure of Miller et al. \cite{-MPX13}, adapted for the radio network model.

\begin{lemma}[Lemma 3.1 of \cite{-HW16}]
\label{lem:clust1}
Let $0 < \beta \le 1$. Any network on $n$ nodes can be partitioned into clusters with strong diameter $O(\frac{\log n}{\beta})$ each with high probability, and every edge cut by this partition with probability $O(\beta)$. This algorithm can be implemented in the radio network setting in $O(\frac{\log^3 n}{\beta})$ rounds.
\end{lemma}

The network being partitioned into clusters means that each node identifies one particular node as its \emph{cluster center}, the subgraph of nodes identifying any particular node as their cluster center is connected, and any node which is a cluster center to anyone must be cluster center to itself. Here the term \emph{``strong diameter''} refers to diameter using only edges within the relevant cluster.

The clustering provided by the application of Lemma \ref{lem:clust1} will be denoted by $\textsc{Partition}(\beta)$.

This framework will be used in our central result, Theorem \ref{thm:cprop} (proven in Appendix \ref{sec:clustering-property}), which states that, upon applying $\textsc{Partition}(\beta)$ with $\beta$ randomly chosen from some range polynomial in $D$, with constant probability the expected distance from some fixed node to its cluster center is $O(\frac{\log n}{\beta\log D})$.

\begin{restatable}{theorem}{cprop}
\label{thm:cprop}
Let $j$ be an integer chosen uniformly at random between $0.01\log D$ and $0.1\log D$, and let $\beta = 2^{-j}$. For any node $v$, with probability at least $0.55$ (over choice of $j$), the expected distance from $v$ to its cluster center upon applying $\textsc{Partition}(\beta)$ is $O(\frac{\log n}{\beta\log D})$.
\end{restatable}

This result applies to the clustering method 
in any setting, not just radio networks, and hence may well be of independent interest. It improves over \cite{-HW16} that expected distance to cluster center is $O(\frac{\log n\log\log n}{\beta\log D})$.

\noindent
This is combined with a means of communicating within clusters from \cite{-GHK13} using the notion of \emph{schedules}.

\begin{lemma}[Lemma 2.1 of \cite{-HW16}]
\label{lem:precomp}
A network of diameter D and n nodes can be preprocessed in $O(D \log^{O(1)} n)$ rounds, yielding a \textbf{schedule} which allows for one-to-all broadcast of k messages in $O(D +k \log n+\log^6 n)$ rounds with high probability. This schedule satisfies the following properties:
\begin{itemize}
\item For some prescribed node $r$, the schedule transmits messages to and from nodes at distance $\ell$ from $r$ in $O(\ell + \log^6 n)$ rounds with high probability.
\item The schedule is periodic of period $O(\log n)$: it can be thought of as restarting every $O(\log n)$ steps.
\end{itemize}
\end{lemma}

Whenever we refer to computing or using \emph{schedules} during our algorithms, we mean using this method. We note that, as shown in Lemma 4.2 of \cite{-HW16}, we can perform this pre-processing in such a way that it succeeds with high probability despite collisions, at a multiplicative $O(\log^{O(1)} n)$ time cost.

\subsection{Algorithm structure}

The general approach then proceeds as follows: first there is a pre-processing phase, in which we partition the network using $\textsc{Partition}(\beta)$ from Lemma \ref{lem:clust1}, and compute schedules within the clusters using Lemma~\ref{lem:precomp}. Then we broadcast the message through the network using these computed schedules within clusters. Any shortest $(u,v)$-path $p$ crosses $O(|p|\beta)$ clusters in expectation, and communication within these clusters takes $O(\frac{\log n}{\beta\log D})$ expected time, so total time required should be $O(|p|\frac{\log n}{\log D}) = O(D\frac{\log n}{\log D})$.

Of course, this omits many of the technical details, and we encounter several difficulties when trying to implement the approach. Firstly, Theorem \ref{thm:cprop} only bounds expected distance to cluster center with constant probability. However, by generating many different clusterings, with different random values of $\beta$, and curtailing application of the schedules after $O(\frac{\log n}{\beta \log D})$ time, we can ensure that we do make sufficient progress with high probability. A second issue is that these values of $\beta$ must somehow be coordinated, which we solve by using an extra layer of ``coarse'' clusters, similarly to \cite{-HW16}. Thirdly, collisions can occur between nodes of different clusters during both precomputation and broadcasting phases. We take several measures to deal with these collisions in our algorithms and analysis.


\subsection{Advances over previous works}

The idea of performing some precomputation locally and then using this local knowledge to perform a global task, occurs frequently in distributed computing. In our setting, the most similar prior work is the $O(D\frac{\log n\log\log n}{\log D}+\log^{O(1)}n )$-time broadcasting algorithm of $\cite{-HW16}$. Here we summarize our main technical differences from that paper and other related works:
\begin{itemize}
\item It was known from \cite{-HW16} that when $\textsc{Partition}(\beta)$ is run with $1/\beta$ randomly selected from a range polynomial in $D$, the expected distance from a node to its cluster center is $O(\frac{\log n\log\log n}{\beta\log D})$. We improve this result with Theorem \ref{thm:cprop}, which states that with constant probability this distance is $O(\frac{\log n}{\beta\log D})$.
\item We demonstrate how, by switching clusterings frequently and curtailing their schedules after $O(\frac{\log n}{\beta\log D})$ time, we can improve the fastest time for broadcasting in radio networks.
\item We show that, with a different method of analysis and an algorithmic background process to deal with collisions, we can extend this method to also complete leader election, a task usually more difficult.
\end{itemize}


\section{Algorithm for \textsc{Compete}}

Since our broadcasting and leader election protocols require the same asymptotic running time and use similar methods (cf. Section \ref{sec:broadcasting-and-leader-election}), we can combine their workings into a single generalized procedure \textsc{Compete}.

\textsc{Compete} takes as input a source set $S$, in which every source $s\in S$ has a message it wishes to propagate, and guarantees, with high probability, that upon completion all nodes know the highest-valued source message. The process takes $O(D\frac{\log n}{\log D}+|S|D^{0.125}+\log^{O(1)}n)$ time (cf. Theorem \ref{thm:compete}), which is within the $O(D\frac{\log n}{\log D}+\log^{O(1)}n)$ time claimed for broadcasting and leader election, as long as $|S| \le D^{0.875}$.

\junk{
As mentioned earlier, it is easy to see how this process generalizes broadcasting: it is simply invoked with the source as the only member of the set $S$. To perform leader election, one can probabilistically generate a set of $\Omega(\log n)$ candidate leaders, and then perform \textsc{Compete} using this set, with $ID$s as the messages to be propagated. Therefore proving the correctness and running time of \textsc{Compete} constitutes the bulk of the work we must do, and the broadcasting and leader election algorithms follow as easy corollaries.
}

Our efficient algorithm for \textsc{Compete} consists of two processes which run concurrently, alternating between steps of each. The main \textsc{Compete} process is designed to propagate messages quickly through most of the network, and the background process is slower, with the purpose of ``papering over the cracks'' in the main process; in this case that means passing messages across coarse cluster boundaries.

\begin{algorithm}[H]
\caption{\textsc{Compete$(S)$}}
\label{alg:C}
\small
\begin{algorithmic}
\State 1) Compute a \emph{coarse clustering} using $\textsc{Partition}(\beta)$ with $\beta = D^{-0.5}$.
\State 2) Compute a schedule within each coarse cluster.
\State 3) Within each coarse cluster, for each integer $j\in [0.01\log D,0.1\log D]$, compute $D^{0.2}$ different \emph{fine clusterings} using $\textsc{Partition}(\beta)$ with $\beta = 2^{-j}$.
\State 4) Compute schedules within all fine clusterings.
\State 5) Each coarse cluster center computes a $D^{0.99}$-length sequence of randomly chosen fine clusterings to use.
\State 6) Transmit this sequence within each coarse cluster, using the coarse cluster schedules.
\State 7) For each fine clustering in the sequence perform \textsc{Intra-Cluster Propagation$(O(\frac{\log n}{\beta\log D}))$} (with the value of $\beta$ corresponding to the fine clustering).
\end{algorithmic}
\end{algorithm}	

In the main process, we first compute a \emph{coarse clustering}, that is, one with comparatively large clusters, which we need to spread shared randomness. Then, within the coarse clusters we compute many different \emph{fine clusterings}, i.e., sub-clusterings with smaller clusters. These are the clusterings we will use to propagate information through the network. The coarse clusters generate and transmit a random sequence of these fine clusterings, which tells their members in what order to use the fine clusterings for this propagation (this was the sole purpose of the coarse clustering). We show that, when applying \textsc{Intra-Cluster Propagation}$(O(\frac{\log n}{\beta\log D}))$ on a clustering with $\beta$ randomly chosen, we have a constant probability of making sufficient progress towards our goal of information propagation. We can treat the progress made during each application of \textsc{Intra-Cluster Propagation} as being independent, since we use a different random clustering each time (and with high probability, whenever we choose a clustering we have used before, we have made sufficient progress in between so that the clusters we are analyzing are far apart and behave independently). Therefore we can use a Chernoff bound to show that with high probability we make sufficient progress throughout the algorithm as a whole.

An issue with the main process, though, is that at the boundaries of the coarse clustering, collisions between coarse clusters can cause  \textsc{Intra-Cluster Propagation} to fail. To rectify this, \emph{we interleave steps of the main process with steps of a background process} (Algorithm \ref{alg:C-B}), e.g., by performing the main process during even time-steps and the background process during odd time-steps.

\begin{algorithm}[H]
\caption{\textsc{Compete$(S)$ - Background Process}}
\label{alg:C-B}
\small
\begin{algorithmic}
\State 1) Compute $D^{0.2}$ different fine clusterings using $\textsc{Partition}(\beta)$ with $\beta = D^{-0.1}$.
\State 2) Compute a schedule within each cluster, for each clustering.
\State 3) Cycling through clusterings in round-robin order, perform \textsc{Intra-Cluster Propagation$(O(\frac{\log n}{\beta}))$}.
\end{algorithmic}
\end{algorithm}	

The background process is simpler: it follows a similar line to the main process, but does not use a coarse clustering, only fine clusterings. This means that we do not have the shared randomness we use in the main process, so we cannot choose $\beta$ randomly (we instead fix $\beta = D^{-0.1}$) and we cannot use a random ordering of fine clusterings (we instead use a round-robin order). As a result, we must run \textsc{Intra-Cluster Propagation} for longer to achieve a constant probability of making good progress, and so the propagation of information is slower (if we were to rely on the background process alone, we would only achieve $O(D\log n + \log^{O(1)} n)$ time).

However, the upside is that there are no coarse cluster boundaries, and so the progress is made consistently throughout the network. Therefore, we can analyze the progress of our algorithm using the faster main process most of the time, and switching to analysis of the background process when the main process reaches a coarse cluster boundary. Since the coarse clusters are comparatively large, their boundaries are reached infrequently, and so we can show that overall the algorithm still makes progress quickly.

Both \textsc{Compete} processes make use of \textsc{Intra-Cluster Propagation} as a primitive, which makes use of the computed clusters and schedule to propagate information. Specifically, the procedure facilitates communication between the cluster center and nodes within $\ell$ hops.

\begin{algorithm}[H]
\caption{\textsc{Intra-Cluster Propagation$(\ell)$}}
\label{alg:ICP}
\small
\begin{algorithmic}
\State 1) Broadcast the highest message known by the cluster center to all nodes within $\ell$ distance.
\State 2) All such nodes which know a higher message participate in a broadcast towards the cluster center.
\State 3) Broadcast the highest message known by the cluster center to all nodes within $\ell$ distance.
\end{algorithmic}
\end{algorithm}	

Here we apply Lemma \ref{lem:precomp}: after computing schedules, it is possible to broadcast between the cluster center and nodes at distance at most $\ell$ in time $O(\ell+\log^{O(1)}n)$. That is, on an outward broadcast all nodes within distance $\ell$ of the cluster center hear its message, and on an inward broadcast the cluster center hears the message of at least one participating node. This would be sufficient in isolation, but since we perform \textsc{Intra-Cluster Propagation} within all fine clusters at the same time, we will describe a background process (Algorithm \ref{alg:ICP-B}) to deal with collisions between fine clusters in the same coarse cluster. As before, we intersperse the steps of the main process and background process, performing one step of each alternately.

\begin{algorithm}[H]
\caption{\textsc{Intra-Cluster Propagation$(\ell)$ - Background Process}}
\label{alg:ICP-B}
\small
\begin{algorithmic}
\State Repeat until main process is complete:
\For {$i = 1$ to $\log n$}
	 \State with probability $2^{-i}$ (coordinated in each cluster) perform one round of \textsc{Decay};
    \State otherwise remain silent for $\log n$ steps.
\EndFor
\end{algorithmic}
\end{algorithm}	

The background process aims to individually inform nodes that border other fine clusters, and therefore may have collisions that prevent them from participating properly in the main process. The goal is to ensure that eventually (we will bound the amount of time that we may have to wait), such a node's cluster will be the only neighboring cluster to perform \textsc{Decay} (Algorithm \ref{alg:D}), which ensures that the node will then hear its cluster's message (with constant probability).

\junk{
The \textsc{Decay} protocol, first introduced by Bar-Yehuda et al. \cite{-BGI92}, is a fundamental primitive employed by many randomized radio network communication algorithms. Its aim is to ensure that, if a node has one or more in-neighbors which wish to transmit a message, it will hear at least one of them.
}
The \textsc{Decay} protocol, first introduced by Bar-Yehuda et al. \cite{-BGI92}, is a fundamental transmission primitive employed by many randomized radio network communication algorithms.

\begin{algorithm}[h]
\caption{\textsc{Decay} at a node $v$}
\label{alg:D}
\small
\begin{algorithmic}
	\State \textbf{for} $i = 1$ \textbf{to} $\log n$, in time step $i$ \textbf{do}: $v$ transmits its message with probability $2^{-i}$
\end{algorithmic}
\end{algorithm}
%
\junk{
The following lemma (cf. \cite{-BGI92}) describes a basic property of \textsc{Decay}, which gives us the desired behavior:
}
\begin{lemma}[\cite{-BGI92}]
\label{corollary:Decay}
After a round of \textsc{Decay}, a node $v$ with at least one participating in-neighbor receives a message with constant probability.
\qed
\end{lemma}


\section{Analysis of \textsc{Compete} Algorithm}
\label{sec:aca}

In this section we prove the following guarantee on the behavior of \textsc{Compete}:

\begin{restatable}{theorem}{compete}
\label{thm:compete}
\textsc{Compete}$(S)$ informs all nodes of the highest message in $S$ within $O(\frac{D\log n}{\log D}+|S|D^{0.125} +\log^{O(1)}n)$ time-steps, with high probability.
\end{restatable}

The precomputation phase of \textsc{Compete}, that is, steps 1--6 of the main process and steps 1-2 of the background process, requires $O(D^{0.99} \log^{O(1)}n) = O(D)$ time, and upon its completion we have all the schedules required to perform \textsc{Intra-Cluster Propagation}. As in \cite{-HW16}, we can ignore collisions during these precomputation steps, since we can simulate each transmission step with $O(\log n)$ rounds of \textsc{Decay} to ensure their success without exceeding $O(D)$ total time.

We first prove a result that allows us to use \textsc{Intra-Cluster Propagation} to propagate messages through the network. During a fixed application of \textsc{Intra-Cluster Propagation}, we call a node \emph{valid} if it can correctly send and receive messages to/from its cluster center despite collisions between fine clusters.

\begin{lemma}
\label{lem:ICP}
For some constant $c$, upon applying \textsc{Intra-Cluster Propagation}$(\ell)$ with $\ell = D^{\Omega(1)}$, a fixed node $u$ at distance at most $\frac{\ell}{c}$ from its cluster center is valid with probability at least $0.99$.
\end{lemma}

\begin{proof}
Let $u$ be a node at distance $d$ from its cluster center, and call nodes on the shortest path from $u$ to the cluster center who border another fine cluster \emph{risky}. We make use of a result of \cite{-HW16} (a corollary of Lemma 3.6 used during proof of Lemma 4.6) which states that any node is risky with probability $O(\beta)$. Therefore the expected number of risky nodes on the path is $O(d\beta)$.

Let $v$ be a risky node bordering $q$ fine clusters, and consider how long $v$ must wait to be informed if it has a neighbor in its own cluster who wishes to inform it. Whenever $2^{-i}$ is within a constant factor of $\frac1q$ during the background process, \textsc{Decay} has $\Omega(\frac1q)$ probability of informing $v$ from its own cluster. This is because with probability $\Omega(\frac1q)$, $v$'s cluster is the only cluster bordering $v$ to perform \textsc{Decay}, and in this case $v$ is informed with constant probability. Since this value of $2^{-i}$ recurs every $O(\log^2 n)$ steps, the time needed to inform $v$ is $O(q\log^2 n)$ in expectation.

We use another result from \cite{-HW16}, Corollary 3.9, which states that with high probability all nodes border $O(\frac{\log n}{\log D}) = O(\log n)$ clusters. Therefore the total amount of time spent informing risky nodes is $O(d\beta \cdot \log^3 n) = O(d)$ in expectation, and since $O(d+\log^{O(1)} n)$ time is required to inform non-risky nodes using the main process, $u$ can communicate with its cluster center in $O(d+\log^{O(1)} n)$ expected time. By choosing sufficiently large $c$, by Markov's inequality $v$ is valid with probability at least $0.99$.
\end{proof}

This will allow us to use \textsc{Intra-Cluster Propagation} to propagate information locally. To make a global argument, we will analyze the \textsc{Compete} algorithm's progress along paths by partitioning said paths into length $D^{0.12}$ subpaths. We call the set of all nodes within $D^{0.11}$ of a subpath its \emph{neighborhood}, and we call a subpath \emph{good} if all nodes in its neighborhood are in the same coarse cluster (and \emph{bad} otherwise). We will show that we pass messages along good subpaths quickly under the main \textsc{Compete} process, and along bad subpaths more slowly under the background process.

To show that there are not too many bad subpaths, we make use of the following result from \cite{-HW16}:

\begin{lemma}[Corollary 3.8 of \cite{-HW16}]
After running $\textsc{Partition}(\beta)$ the probability of a fixed node u having nodes from t distinct clusters at distance d or less from u is at most $(1-e^{-\beta(2d+1)})^{t-1}$.
\qed
\end{lemma}

Therefore the probability of a node $u$ having nodes from two different coarse clusters within $D^{0.11}$ distance is at most $1-e^{-D^{-0.5}(2D^{0.11}+1)} \le 1-e^{-3D^{-0.39}} \le 3D^{-0.39}$. Taking the union bound over all nodes in a path, we find that any length-$D^{0.12}$ path is bad with probability at most $D^{0.12} \cdot 3D^{-0.39} \le D^{-0.26}$.

\begin{lemma}
\label{lem:badp}
All shortest paths $p$ between two vertices have $O(D^{0.63})$ bad subpaths, with high probability.
\end{lemma}

\begin{proof}
Fix some shortest path $p$. As in the proof of Lemma 4.3 of \cite{-HW16}, we first condition on the event that all exponentially distributed random variables $\delta_v$ used when computing the coarse clustering are $\le 2D^{0.5}\log n$, which is the case with high probability (for details of how the clustering algorithm works see Appendix \ref{sec:clustering-property}). Then, the events that two length-$D^{0.12}$ subpaths of distance at least $5D^{0.5}\log n$ apart are bad are independent, since they are not affected by any of the same $\delta_v$. If we label the length-$D^{0.12}$ subpaths of $p$ in order from one end of the path to the other, and group them by label mod $6D^{0.38}\log n$, then the badness of every subpath is independent from all the others in its group. Hence, the number of bad subpaths in each group is binomially distributed, and is $O(\frac{D}{D^{0.12} \cdot 6D^{0.38}\log n}\cdot D^{-0.26}) = O(D^{0.24})$ with high probability by a Chernoff bound. By the union bound over all of the groups, the total number of bad subpaths is $O(D^{0.62})$ with high probability. If we allow this amount to be as high as $O(D^{0.63})$, we can reduce the probability that we exceeds it to $n^{-c}$ for an arbitrarily large constant $c$. We can then take a union bound over all $n^2$ shortest paths, and find that they all have $O(D^{0.63})$ bad subpaths with high probability.
\end{proof}

Having bounded the number of bad subpaths, we can show we can pass messages along them using the background process, quickly enough that we do not exceed the algorithm's stated running time in total. Note that here, and henceforth, we will refer to messages by their place in increasing order out of all messages of nodes in $S$. That is, by \emph{message $j$} we mean the $j^{th}$ highest message in $S$.

\begin{lemma}[\textbf{\emph{Bad subpaths}}]
\label{lem:BS}
Let $p$ be any $(u,v)$-path of length at most $D^{0.12}$. Let $j$ be the minimum, over all nodes $v$ in $p$'s neighborhood, of the highest message known by $v$ at time-step $t$. If, at timestep $t$, $u$ knows a message higher than $j$, then by time-step $t' = t + O(D^{0.121})$ all nodes in $p$ know a message at least as high as $j+1$ with high probability.
\end{lemma}

\begin{proof}
We analyze only the background process, and consider separately each fine clustering used in the sequence between time-steps $t$ and $t'$. For any such clustering, let $w$ be the furthest node along $p$ which knows a message at least as high as $j+1$. We call the clustering \emph{good} if:
\begin{itemize}
\item all nodes in $w$'s cluster are $O(D^{0.1}\log n)$ distance from the cluster center;
\item the node $x$ which is $\frac{D^{0.1}}{c}$ nodes along $p$ from $w$ is in its cluster as $w$;
\item $x$ and $w$ are valid (recall that this means they succeed in \textsc{Intra-Cluster Propagation}).
\end{itemize}

By Lemma \ref{lem:clust1} the first event occurs with high probability, by Corollary 3.7 of \cite{-HW16}, we can make the probability of the second event an arbitrarily high constant by our choice of $c$, and by Lemma \ref{lem:ICP} and the union bound, the third event occurs with probability at least $1-2(1-0.99) = 0.98$, conditioned on the first. Therefore the clustering is good with probability at least $\frac 12$, by applying the union bound again.

By a Chernoff bound, $\Omega(D^{0.02})$ of the clusterings applied between times $t$ and $t'$ will be good. Consider each good clustering in turn. After applying such a clustering, $w$'s cluster will be informed of an ID higher than $j$. Every time this occurs, $w$ advances at least $\frac{D^{0.1}}{c}$, and so by time $t'$ the entire path knows a message at least as high as $j+1$.
\end{proof}

We now make a similar argument for the good subpaths, but since we can use the main \textsc{Compete} process without fear of collisions from other coarse clusters, we get a better time bound:

\begin{lemma}[\textbf{\emph{Good subpaths}}]\label{lem:GS}
Let $p$ be any good $(u,v)$-path of length at most $D^{0.12}$. Let $j$ be the minimum, over all nodes $v$ within $D^{0.11}$ distance $p$, of the highest message known by $v$ at time-step $t$. If, at timestep $t$, $u$ knows a message higher than $j$, then by time-step $t' = t + O(D^{0.12}\frac{\log n}{\log D})$ all nodes in $p$ know a message at least as high as $j+1$ with high probability.
\end{lemma}

\begin{proof}
We analyze only the main procedure, and consider separately each fine clustering used in the sequence between time-steps $t$ and $t'$. For any such clustering, let $w$ be the furthest node along $p$ which knows a message at least as high as $j+1$. We call the clustering \emph{good} if:
\begin{itemize}
\item $w$ is at distance at most $c_1\frac{\log n}{\beta \log D}$ from its cluster center;
\item the node $x$ which is $\frac{D^{0.1}}{c}$ nodes along $p$ from $w$ is in the same cluster as $w$;
\item $x$ and $w$ are valid (recall that this means they succeed in \textsc{Intra-Cluster Propagation}).
\end{itemize}

By Theorem \ref{thm:cprop}, and using Markov's inequality, we can choose $c_1$ such that the first event occurs with probability at least $0.54$, conditioned on all previous randomness. By Corollary 3.7 of \cite{-HW16}, we can choose $c_2$ so that the second event occurs with probability at least $0.99$, also conditioned on all previous randomness. By Lemma \ref{lem:ICP} the probability that $x$ and $w$ are valid, conditioned on the first event, is at least $0.98$. Therefore each fine clustering is good with probability at least $\frac 12$ (by the union bound).

Let $S$ be the set of all clusterings applied between time-steps $t$ and $t'$. We are interested in the quantity $\sum_{s \in S \text{ is good}}\beta_s^{-1}$. Note that this majorizes the quantity $\sum_{s \in S}x_s$, where the $x_s$ are independent Bernoulli variables which take value $\beta_s^{-1}$ with probability $\frac 12$ and $0$ otherwise. The expected value of this quantity is $\frac 12 \sum_{s \in S \text{ is good}}\beta_s^{-1} \ge \frac{c}{3} D^{0.12}$. By Hoeffding's inequality,
\begin{displaymath}
    \Prob{\sum_{s \in S} x_s
        \le
    \frac{c}{6} D^{0.12}}
        \le
    e^{-\frac{2|S|^2(\frac{c}{6} D^{0.12})^2}{\sum_{s \in S}\beta_s^{-2}}}
        \le
    e^{-\frac{2|S|(\frac{c}{6} D^{0.12})^2}{D^{0.1}}}
        \le
    e^{-\log^2 n}
        \enspace.
\end{displaymath}

By time $t'$, $w$ has advanced at least $\frac{\sum_{s \in S}}{c_2} \ge \frac{cD^{0.12}}{6c_2}$ steps along $p$, and so by choosing a sufficiently large constant in the big-Oh notation for $t'$, we can ensure that every node in $p$ knows a message at least as high as $j+1$.
\end{proof}

\noindent We combine these results to show how to propagate messages along any shortest path between two nodes.

\begin{lemma}[\textbf{\emph{All shortest paths}}]
\label{lem:asp}
Let $u$ and $v$ be any nodes in $\net$, $p$ be some shortest $(u,v)$-path, and let $b$ be the number of bad length-$D^{0.12}$ subpaths of $p$. If $u$ knows a message at least as high as $i$ at time-step~$t$, then after $t+2k(\frac{|p|\log n}{\log D} + (2 i+b) D^{0.125})$ steps, $v$ knows a message at least as high as $i$ with high probability.
\end{lemma}

\begin{proof}
We prove the lemma using double induction. Our `outer' induction shall be on the value $i$.

\noindent\textbf{Base case: $i = 1$.}\quad
By Lemma \ref{lem:badp}, $p$ contains at most $D^{0.63}$ bad subpaths. Applying Lemmas \ref{lem:BS} and \ref{lem:GS}, the time taken to inform $v$ of a message at least as high as $1$ is at most $D^{0.63} \cdot O(D^{0.121}) + D^{0.88}\cdot O(D^{0.12}\frac{\log n}{\log D})\le kD\frac{\log n}{\log D}$.

\noindent\textbf{Inductive step:}\quad
We can now assume the claim for $i=\ell-1$ (inductive assumption 1), and prove the inductive step $i=\ell$. We do this using a second induction, on $|p|$.

\textbf{Base case: $|p|\le D^{0.12}$.}\quad
$p$ is a single subpath. If $p$ is good, then by inductive assumption 1, all nodes within $D^{0.11}$ of $p$ know an ID at least as high as $\ell-1$ by time-step $t+k(\frac{(|p|+D^{0.11})\log n}{\log D} + (2 (\ell-1)+1) D^{0.125})$. Then, by Lemma \ref{lem:GS}, $v$ knows an ID at least as high as $\ell$ by time-step
\begin{displaymath}
    t +
    k \left(\frac{(|p|+D^{0.11})\log n}{\log D} + (2 \ell-1) D^{0.125}\right)
    +
    cD^{0.12}\frac{\log n}{\log D}
        \le
    t + k \left(\frac{|p|\log n}{\log D} + 2 \ell D^{0.125}\right)
        \enspace.
\end{displaymath}

If $p$ is bad then by inductive assumption (1), all nodes within $D^{0.11}$ of $p$ know an ID at least as high as $\ell-1$ by time-step $t+k(\frac{(|p|+D^{0.11})\log n}{\log D} + (2 (\ell-1)+2) D^{0.125})$. Then, by Lemma \ref{lem:BS}, $v$ knows an ID at least as high as $i$ by time-step
\begin{displaymath}
    t +
    k \left(\frac{(|p|+D^{0.11})\log n}{\log D} + 2 \ell D^{0.125}\right)
    + cD^{0.121}
        \le
    t + k \left(\frac{|p|\log n}{\log D} + (2 \ell+1) D^{0.125}\right)
        \enspace.
\end{displaymath}

\textbf{Inductive step:}\quad
Having proved the base case, we can now assume the claim for $i=\ell$ and $|p|< q$ (inductive assumption 2), and prove the inductive step $|p|=q$.

Let $u'$ be the start node of the last sub-path of $p$. If this subpath is good, then by inductive assumption 2, $u'$ knows an ID at least as high as $\ell$ by time-step $t+k(\frac{(|p|-D^{0.12})\log n}{\log D} + (2i+b)  D^{0.125})$. By inductive assumption (1), all nodes within $D^{0.11}$ of $p$ know a message at least as high as $\ell-1$ by time-step $t+k({(|p|+D^{0.11})\log n}{\log D} + (2 (\ell-1)+(b+1))) D^{0.125} \le t+k(\frac{(|p|-D^{0.12})\log n}{\log D} + (2 \ell+b)) D^{0.125}$ Therefore, by Lemma \ref{lem:GS}, $v$ knows a message at least as high as $\ell$ by time-step \begin{align*}
    t +
    k \left(\frac{(|p|-D^{0.12})\log n}{\log D} + (2 \ell+b) D^{0.125}\right)
    + cD^{0.12}\frac{\log n}{\log D}
        \le
    t + k \left(\frac{|p|\log n}{\log D} + (2 \ell+b) D^{0.125}\frac{\log n}{\log D}\right)
    \enspace.
\end{align*}

If the subpath is bad, then by inductive assumption (2), $u'$ knows an ID at least as high as $\ell$ by time-step $t+k(\frac{(|p|-D^{0.12})\log n}{\log D} + (2\ell+b-1)  D^{0.125})\le t+k(\frac{(|p|+D^{0.11})\log n}{\log D} + (2 (\ell-1)+b) D^{0.125})$. By inductive assumption (1), all nodes within $D^{0.11}$ of $p$ know a message at least as high as $\ell-1$ by time-step $t+k(\frac{(|p|+D^{0.11})\log n}{\log D} + (2 (\ell-1)+b) D^{0.125}) $ Therefore, by Lemma \ref{lem:BS}, $v$ knows a message at least as high as $\ell$ by time-step
\begin{align*}
    t +
    k \left(\frac{(|p|+D^{0.11})\log n}{\log D} + (2 (\ell-1)+b) D^{0.125}\right)
    + cD^{0.121}
        \\
        \le
    t + k \left(\frac{|p|\log n}{\log D} + (2 \ell+b) D^{0.125}\frac{\log n}{\log D}\right)
    \enspace.
\end{align*}

This completes the proof by induction.
\end{proof}

We are now ready to prove Theorem \ref{thm:compete}:

\begin{proof}
The precomputation phase takes at most $O(D+\log^{O(1)})$ time. Upon beginning the \textsc{Intra-Cluster Propagation} phase, one node $u$ knows the highest message. Therefore by Lemma \ref{lem:asp}, all nodes $v$ know this message within $2k(\frac{|dist(u,v)|\log n}{\log D} + (2 |S|+b) D^{0.125}) = O(\frac{D\log n}{\log D}+|S|D^{0.125} +\log^{O(1)}n)$ time-steps, with high probability.
\end{proof}


\section{Application of \textsc{Compete} to broadcasting and leader election}
\label{sec:broadcasting-and-leader-election}

It is not difficult to see that \textsc{Compete} can be used to perform both broadcasting and leader election.

\junk{
\begin{algorithm}[H]
\caption{Broadcasting}
\label{alg:BC}
\begin{algorithmic}
\State 1) Perform \textsc{Compete}$(\{s\})$.
\end{algorithmic}
\end{algorithm}	
}

\begin{theorem}
\label{thm:broadcasting}
\textsc{Compete}$(\{s\})$ completes broadcasting in $O(D\frac{\log n}{\log D}+\log^{O(1)} n)$ time with high probability.
\end{theorem}

\begin{proof}
\textsc{Compete} informs all nodes of the highest message in the message set in time $O(D\frac{\log n}{\log D}+\log^{O(1)} n)$, with high probability. Since this set contains only the source message, broadcasting is completed.
\end{proof}

\begin{algorithm}[H]
\caption{{\sc Leader Election}}
\label{alg:LE}
\small
\begin{algorithmic}
\State 1) Nodes choose to become candidates $\in C$ with probability $\frac{100 \log n}{n}$.
\State 2) Candidates randomly generate $\Theta(\log n)$-bit IDs.
\State 3) Perform \textsc{Compete}$(C)$.
\end{algorithmic}
\end{algorithm}	

\begin{theorem}
\label{thm:leader-election}
Algorithm \ref{alg:LE} completes leader election within $O(D\frac{\log n}{\log D}+\log^{O(1)} n)$ time with high probability
\end{theorem}

\begin{proof}
With high probability $|C| = \Theta(\log n)$ and all candidate IDs are unique. Conditioning on this, \textsc{Compete} informs all nodes of the highest candidate ID within time $O(D\frac{\log n}{\log D}+\log^{O(1)} n)$, with high probability. Therefore leader election is completed.
\end{proof}


\section{Conclusions}

The tasks of broadcasting and leader election in radio networks are longstanding, fundamental problems in distributed computing. Our main contribution are new algorithms for these problems that improve running times for both to $O(D\frac{\log n}{\log D}+\log^{O(1)}n)$, which is optimal for a wide range of $D$.

There is no better lower bound than $\Omega(D+\log^2 n)$ for broadcasting or leader election when spontaneous transmissions are allowed, so the most immediate open question is to close that gap. While a tighter analysis of our method might trim the additive polylog$(n)$ term significantly, it is difficult to see how $\log^2 n$ could be reached without a radically different approach. Similarly, the $D\frac{\log n}{\log D}$ term seems to be a limit of the clustering approach, and reducing it to $D$ would likely require significant changes. In fact, we would not be surprised if our upper bounds $O(D \frac{\log n}{\log D})$ were tight for $D = \Omega(\log^c n)$ for a sufficiently large constant $c$.

The main focus of this paper has been to study the impact of spontaneous transmissions for basic communication primitives in randomized algorithms undirected networks. An interesting question is whether spontaneous transmissions can help in \emph{directed} networks, which would be very surprising, or for \emph{deterministic} protocols.


\newcommand{\Proc}{Proceedings of the\xspace}

\newcommand{\STOC}{Annual ACM Symposium on Theory of Computing (STOC)}
\newcommand{\FOCS}{IEEE Symposium on Foundations of Computer Science (FOCS)}
\newcommand{\SODA}{Annual ACM-SIAM Symposium on Discrete Algorithms (SODA)}

\newcommand{\COCOON}{Annual International Computing Combinatorics Conference (COCOON)}
\newcommand{\DISC}{International Symposium on Distributed Computing (DISC)}
\newcommand{\ESA}{Annual European Symposium on Algorithms (ESA)}
\newcommand{\ICALP}{Annual International Colloquium on Automata, Languages and Programming (ICALP)}
\newcommand{\IPL}{Information Processing Letters}
\newcommand{\JACM}{Journal of the ACM}
\newcommand{\JALGORITHMS}{Journal of Algorithms}
\newcommand{\JCSS}{Journal of Computer and System Sciences}
\newcommand{\PODC}{Annual ACM Symposium on Principles of Distributed Computing (PODC)}
\newcommand{\SICOMP}{SIAM Journal on Computing}
\newcommand{\SPAA}{Annual ACM Symposium on Parallelism in Algorithms and Architectures}
\newcommand{\STACS}{Annual Symposium on Theoretical Aspects of Computer Science (STACS)}
\newcommand{\TALG}{ACM Transactions on Algorithms}
\newcommand{\TCS}{Theoretical Computer Science}
\bibliography{spont}


\clearpage
\appendix
\begin{center}
\bf\huge Appendix
\end{center}


\section{Clustering property (proof of Theorem \ref{thm:cprop})}
\label{sec:clustering-property}

In this section we prove a key property of the clustering method in our algorithm \textsc{Partition$(\beta)$}. Our analysis is based on a method first introduced in \cite{-MPX13}.
The concept is as follows: each node $v$ independently generates an exponentially distributed random variable $\delta_v$, that is, a variable taking values in $\mathbb{R}_{\ge 0}$  with $\prob{\delta_v \le y} = 1-e^{-\beta y}$. Then, each node chooses its cluster center $u$ to be the node maximizing $\delta_u - dist(u,v)$. It can be seen by the triangle inequality that a node which is cluster center to any node is also cluster center to itself.
For details of how to implement this in the radio network setting, see \cite{-HW16}.

The rest of this section is now devoted to the proof of Theorem \ref{thm:cprop}.

\cprop*

Our first step in proving Theorem \ref{thm:cprop} is to obtain a bound for distance to cluster center which is based upon the number of nodes at each distance layer from $v$. To this purpose, let $A_i(v)$ be the set of nodes at distance $i$ from $v$ and denote $x_i = |A_i(v)|$. Denote $\xx \in \mathbb{N}_0^D$ to be the vector with these $x_i$ as coefficients.

Denote $T_{\xx,\beta} = \sum_{i=0}^{D} i x_i e^{-i\beta}$ and $B_{\xx,\beta} = \sum_{i=0}^{D} x_i e^{-i\beta}$. Denote $S_{\xx,\beta} = \frac{T_{\xx,\beta}}{B_{\xx,\beta}} = \frac{\sum_{i=0}^{D} i x_i e^{-i\beta}}{\sum_{i=0}^{D} x_i e^{-i\beta}}$. These quantities will be used in the following auxiliary lemma describing the expected distance from any fixed $v$ to its cluster center after applying $\textsc{Partition}(\beta)$.

\begin{lemma}
\label{lem:edist}
For any fixed node $v$ and value $\beta$ with $D^{-0.01}\le\beta\le D^{-0.1}$, the expected distance from $v$ to its cluster center upon applying $\textsc{Partition}(\beta)$ is at most $\frac{5 \sum_{i=0}^{D} i x_i e^{-i\beta}}{\sum_{i=0}^{D} x_i e^{-i\beta}} = 5S_{\xx,\beta}$.
\end{lemma}

\begin{proof}
We bound expected distance to cluster center:
\begin{align*}
    \expect{\text{distance from $v$ to its cluster center}}
        &=
    \sum_{i=0}^{D}i\prob{\text{$v$'s cluster center is distance $i$ away}}
        \\
        &=
    \sum_{i=0}^{D} i \cdot \left(
        \sum_{u \in A_i(v)}\prob{\text{$u$ is $v$'s cluster center}}
        \right)
        \enspace.
\end{align*}
\junk{
\begin{align*}
    \expect{\text{distance }&\text{from $v$ to its cluster center}}
        =
    \sum_{i=0}^{D}i\prob{\text{$v$'s cluster center is distance $i$ away}}
        \\
        &=
    \sum_{i=0}^{D} i x_i\prob{\text{some fixed node $u$ at distance $i$ is $v$'s cluster center}}
        \enspace.
\end{align*}
}

We concentrate on this latter probability and henceforth fix $u \in A_i(v)$ to be some node at distance $i$ from $v$. We note that
\begin{displaymath}
    \prob{\text{$u$ is $v$'s cluster center}} =
    \int_i^{\infty} \beta e^{-\beta p}
        \prob{\text{$u$ is $v$'s cluster center}|\delta_u = p} dp
\end{displaymath}
by conditioning on the value of $\delta_u$ over its whole range and multiplying by the corresponding probability density function (we can start the integral at $i$ since if $\delta_u<i$ the probability of $u$ being $v$'s cluster center is~$0$).

Having conditioned on $\delta_u$, we can evaluate the probability that $u$ is $v$'s cluster center based on the random variables generated by other nodes, since the probabilities that each other node `beats' $u$ are now independent:
\begin{displaymath}
    \prob{\text{$u$ is $v$'s cluster center}} =
            \int_i^{\infty} \beta e^{-\beta p} \prod_{w\neq u}
            \prob{\delta_w - dist(v,w)< \delta_u - dist(v,u)|\delta_u = p} dp
        \enspace.
\end{displaymath}

We can simplify by grouping the nodes $w$ based on distance from $v$, though we must be careful to include a $\frac{1}{\prob{\delta_u < p }}$ term to cancel out $u$'s contribution to the resulting product:
\begin{displaymath}
    \prob{\text{$u$ is $v$'s cluster center}} =
        \int_i^{\infty} \frac{\beta e^{-\beta p}}{\prob{\delta_u < p }} \prod_{k = 0}^{D} \prod_{w\in A_k(v)} \prob{\delta_w - k< p - i} dp
        \enspace.
\end{displaymath}

Plugging in the cumulative distribution function gives the following:
\begin{displaymath}
    \prob{\text{$u$ is $v$'s cluster center}} =
        \int_i^{\infty} \frac{\beta e^{-\beta p}}{1-e^{-\beta p}}
            \prod_{k = 0}^{D}\prod_{w\in A_k(v)}1-e^{-\beta ( p - i + k)} dp
        \enspace.
\end{displaymath}

We use the standard inequality $1-y\le e^{-y}$ for $y \in [0,1]$, here setting $y = e^{-\beta ( p - i + k)}$ (note that since $i \le p$, we have $y \in [0,1]$), and account for the second product by taking the contents to the power of $x_k$:
\begin{displaymath}
    \prob{\text{$u$ is $v$'s cluster center}}
        \le
    \int_i^{\infty}\!\!\frac{\beta e^{-\beta p}}{1-e^{-\beta p}}
        \prod_{k = 0}^D \prod_{w\in A_k(v)}e^{-e^{-\beta  (p - i + k)}} dp
        =
    \int_i^{\infty}\!\!\frac{\beta e^{-\beta p}}{1-e^{-\beta p}}
        \prod_{k = 0}^D e^{-e^{-\beta  (p - i + k)}x_k} dp
        \enspace.
\end{displaymath}

We can also remove the remaining product by taking it as a sum into the exponent, and re-arranging some terms yields:
\begin{displaymath}
    \prob{\text{$u$ is $v$'s cluster center}}
        \le
    \int_i^{\infty} \frac{\beta e^{-\beta p}}{1-e^{-\beta p}}
        e^{-e^{\beta(i- p)}\sum_{k = 0}^{D}x_k e^{-\beta k}} dp
        =
    \int_i^{\infty} \frac{\beta e^{-\beta p}}{1-e^{-\beta p}}
        e^{-e^{\beta(i- p)}B_{\xx,\beta}} dp
        \enspace,
\end{displaymath}
where for succinctness we use our definition $B_{\xx,\beta} = \sum_{i=0}^{D} x_i e^{-i\beta}$.

At this point we split the integral and bound the parts separately, since they exhibit different behavior:
\begin{displaymath}
    \prob{\text{$u$ is $v$'s cluster center}} \le J+K
        \enspace,
\end{displaymath}
where $J = \int_i^{\frac 1\beta} \frac{\beta e^{-\beta p}}{1-e^{-\beta p}} e^{-e^{\beta(i- p)}B_{\xx,\beta}} dp$ and $K = \int_{\frac 1\beta}^{\infty} \frac{\beta e^{-\beta p}}{1-e^{-\beta p}} e^{-e^{\beta(i- p)}B_{\xx,\beta}} dp$.

To bound $J$, we make use of the following bound on $B_{\xx,\beta}$:
\begin{displaymath}
    B_{\xx,\beta}
        =
    \sum_{k=0}^D x_k e^{-k\beta}
        \ge
    \sum_{k=0}^{\lceil\frac D2\rceil} e^{-k\beta}
        \ge
    \int_{-1}^{\frac D2} e^{-z\beta} dz
        =
    \frac{-1}{\beta}(e^{-\frac {\beta D}{2}}-e^{-\beta})
        \ge
    \frac {1}{2\beta}
        \enspace.
\end{displaymath}

So,
\begin{displaymath}
    J
        =
    \int_i^{\frac 1\beta} \frac{\beta e^{-\beta p}}{1-e^{-\beta p}}
            e^{-e^{\beta(i- p)}B_{\xx,\beta}} dp
        \le
    \int_i^{\frac 1\beta} \frac{\beta e^{-\beta p}}{1-e^{-\beta p}}
            e^{-e^{\beta(i- p)}\frac{1}{2\beta}} dp
        \enspace.
\end{displaymath}

Since $e^{\beta(i- p)}\ge e^{-1}$, we obtain,
\begin{displaymath}
    J
        \le
    \int_1^{\frac 1\beta} \frac{\beta e^{-\beta p}}{1-e^{-\beta p}}
            e^{-\frac{1}{2e\beta}} dp
        =
    \beta e^{-\frac{1}{2e\beta}}
            \int_1^{\frac 1\beta} \frac{e^{-\beta p}}{1-e^{-\beta p}} dp
        \enspace.
\end{displaymath}

We can then use that $\int_{a}^{b} \frac{e^{-\beta p}}{1-e^{-\beta p}} = \frac{1}{\beta}\log \frac{(1-e^{\beta b})}{(1-e^{\beta a})}+a-b$ to evaluate:
\begin{displaymath}
    J
        \le
    e^{-\frac{1}{2e\beta}} \log \frac{(1-e)}{(1-e^{\beta})}
        \enspace.
\end{displaymath}

Since $e^\beta > 1+\beta$, re-arranging yields,
\begin{displaymath}
    J
        \le
    e^{-\frac{1}{2e\beta}} \log \frac{e-1}{\beta}
        \enspace.
\end{displaymath}

Finally, since we can assume that $\frac 1 \beta \ge \log^c n$ for some sufficiently large $c$, we obtain,
\begin{displaymath}
    J
        \le
    e^{-\frac{\log^2 n}{2e}} \log \frac{e-1}{\beta}
        \le
    n^{-2}
        \enspace.
\end{displaymath}

We now turn our attention to $K = \int_{\frac 1\beta}^{\infty} \frac{\beta e^{-\beta p}}{1-e^{-\beta p}} e^{-e^{\beta(i- p)}B_{\xx,\beta}} dp$. Since $1-e^{-\beta p}\ge 1-e^{-1} > \frac12$, we get
\begin{displaymath}
    K
        <
    \int_{\frac{1}{\beta}}^{\infty} 2\beta e^{-\beta p}
        e^{-e^{ -\beta p}e^{\beta i}B_{\xx,\beta}} dp
        \enspace.
\end{displaymath}

Using that $e^{-e^{-\beta p}} \le 1-\frac12 e^{-\beta p}$ (since $0 \le e^{-\beta p} \le 1$), we obtain,
\begin{displaymath}
    K
        <
    \int_{\frac 1\beta}^{\infty} 2\beta e^{-\beta p}
        (1-\frac 12 e^{ -\beta p})^{e^{\beta i}B_{\xx,\beta}} dp
        \enspace.
\end{displaymath}

Evaluating the integral, using $\int_a^\infty e^{-\beta p}(1-\frac 12 e^{-\beta p})^c = \frac{(e^{-a\beta}-2)(1 - \frac12 e^{-a \beta})^c +2}{\beta (1+c)}$, we obtain,
\begin{displaymath}
    K
        <
    2 \frac{(e^{-1}-2)(1 - \frac 12 e^{-1})^{e^{\beta i}B_{\xx,\beta}}+2}{ 1+e^{\beta i}B_{\xx,\beta}}
        \le
    \frac{4}{e^{\beta i}B_{\xx,\beta}}
        \enspace.
\end{displaymath}

We can now combine our calculations to prove the lemma. Since $x_i = |A_i(v)|$, we have,
\begin{align*}
    \expect{\text{distance }&\text{from $v$ to its cluster center}}
        =
    \sum_{i=0}^D i \sum_{u\in A_i(v)}\prob{\text{$u$ is $v$'s cluster center}}
        \\
        &\le
    \sum_{i=0}^D i x_i (J+K)
        <
    \sum_{i=0}^D i x_i \left(n^{-2}+\frac{4}{e^{\beta i}B_{\xx,\beta}}\right)
        \\
        &\le
        n^{-2}\sum_{i=0}^D D x_i  +
            \frac{4\sum_{i=0}^D i x_ie^{-\beta i}}{B_{\xx,\beta}}
        \le
    \frac{D}{n} + 4S_{\xx,\beta}
        \le
    5S_{\xx,\beta}
        \enspace.
    \qedhere
\end{align*}
\end{proof}


\subsection{Simplifying the form of $\xx$ to bound $S_{\xx,\beta}$}
\label{subsec:clustering-property-simplification-of-x}

To continue with the proof of Theorem \ref{thm:cprop}, by Lemma \ref{lem:edist}, our main goal is to upper bound the value of $S_{\xx,\beta} = \frac{\sum_{i=0}^{D} i x_i e^{-i\beta}} {\sum_{i=0}^{D} x_i e^{-i\beta}}$. To simplify our analysis, we will apply two transformations to $\xx$ which will provide us with useful properties for bounding, while not increasing any $S_{\xx,\beta}$ by more than a constant factor.

The first transformation we apply will be to collate coefficients of $\xx$ into indices which are just the powers of $2$. That is, we sum the coefficients of $\xx$ over regions of doubling size.

Let $f: \mathbb{R}^{D+1} \rightarrow \mathbb{R}^{D+1}$ be given by
$f(\xx)_i = \begin{cases}
    \sum_{\ell = 2i}^{4i-1} x_\ell &
        \text{if $i=2^k$ for some $k \in \mathbb{N}_0$,}\\
    0 &\text{otherwise.}
\end{cases}$

The second transformation is to ensure that the coefficients of $\xx$, which we now require to be $0$ at all indices which are not powers of $2$, are not ``too decreasing;'' in particular, we guarantee that each non-zero coefficient is at least half the previous one.

Let $g: \mathbb{R}^{D+1} \rightarrow \mathbb{R}^{D+1}$ be given by
$g(\xx)_i = \begin{cases}
    \sum_{\ell \le i} \frac{\ell x_\ell}{i} &
        \text{if $i = 2^k$ for some $k\in \mathbb{N}_0$,}\\
    0  &\mbox{otherwise.}
\end{cases}$

Having performed these transformations, we can make an argument about the ratios of consecutive non-zero coefficients to bound $S_{x,\beta}$.


\subsubsection{Bounding $S_{\xx,\beta}$ in terms of $S_{f(\xx),\beta}$}
\label{subsubsec:clustering-property-1st-simplification-of-x}

We begin with bounding $S_{\xx,\beta}$ in terms of $S_{f(\xx),\beta}$.

\begin{claim}
\label{cl:trans1}
For all $\xx \in \mathbb{N}_0^D$, $S_{\xx,\beta} \le 11 S_{f(\xx),\beta}$.
\end{claim}

\begin{proof}
We start with the following auxiliary claim.

\begin{claim}
\label{cl:avg}
Consider an expression of the form $\frac{\sum_{i=0}^D i w_i}{\sum_{i=0}^D w_i}$, where all $w_i$ are non-negative. Let $p$ be an integer with $p < \frac{\sum_{i=0}^D i w_i}{\sum_{i=0}^D w_i}$. For all $i < p$ let $0 \le w'_i \le w_i$, and for all $i \ge p$ let $w'_i \ge w_i$. Then $\frac{\sum_{i=0}^D i w'_i}{\sum_{i=0}^D w'_i} > p$.
\end{claim}

Intuitively, consider $\frac{\sum_{i=0}^D i w_i}{\sum_{i=0}^D w_i}$ as a weighted average of the $i$ (with weights $w_i$). The claim then says that for any $p$ which is less than the value of the average, increasing the weights for indices higher than $p$ and reducing them for indices lower than $p$ cannot reduce the weighted average below $p$.

\begin{proof}[Proof of Claim \ref{cl:avg}]
\begin{align*}
\frac{\sum_{i=0}^{D}iw'_i}{\sum_{i=0}^{D}w'_i}
    &=
    \frac{\sum_{i=0}^{D}iw_i + \sum_{i=0}^{D}i(w'_i-w_i)}{\sum_{i=0}^{D}w'_i}\\
    &= \frac{\frac{\sum_{i=0}^{D}iw_i}{\sum_{i=0}^{D}w_i} \cdot \sum_{i=0}^{D}w_i + \sum_{i=0}^{p-1}i(w'_i-w_i)+ \sum_{i=p}^{D}i(w'_i-w_i)}{\sum_{i=0}^{D}w_i+\sum_{i=0}^{D}(w'_i-w_i)}
    \\
    &
    > \frac{p\cdot \sum_{i=0}^{D}w_i + \sum_{i=0}^{p-1}p(w'_i-w_i)+ \sum_{i=p}^{D}p(w'_i-w_i)}{\sum_{i=0}^{D}w_i+\sum_{i=0}^{D}(w'_i-w_i)}\\
    &= \frac{p\cdot \left(\sum_{i=0}^{D}w_i + \sum_{i=0}^{D}(w'_i-w_i)\right)}{\sum_{i=0}^{D}w_i+\sum_{i=0}^{D}(w'_i-w_i)}
    =p
\enspace.
\qedhere
\end{align*}
\end{proof}

We apply Claim \ref{cl:avg} to analyze the effect of the transformation $f$, in particular to compare $S_{f(\xx),\beta}$ with $S_{\xx,\beta}$. First we find an expression for $S_{\xx,\beta}$ in a form for which we can use the claim:
\begin{align*}
    S_{\xx,\beta}
    &
    = \frac{\sum_{i=0}^{D} i x_i e^{-i\beta}}{\sum_{i=0}^{D} x_i e^{-i\beta}}
    = \frac{\sum_{i=0}^{D} i w_i}{\sum_{i=0}^{D} i w_i }
    \enspace,
\end{align*}
where $w_i = x_i e^{-i\beta}$.

Next we do the same for $S_{f(\xx),\beta}$:

\begin{align*}
    S_{f(\xx),\beta}
    &
    = \frac{\sum_{k=0}^{\log D} 2^k \sum_{\ell=2^{k+1}}^{2^{k+2}-1} x_\ell e^{-2^k \beta}}{\sum_{k=0}^{\log D} \sum_{\ell=2^{k+1}}^{2^{k+2}-1}x_\ell e^{-2^k \beta}}
    = \frac{\sum_{\ell=2}^{D} 2^{\lfloor\log \ell - 1\rfloor} x_\ell e^{-2^{\lfloor\log \ell - 1\rfloor}\beta}}{\sum_{\ell=2}^{D} x_\ell e^{-2^{\lfloor\log \ell - 1\rfloor}\beta}}
        \enspace.
\end{align*}

We multiply both the numerator and denominator by a scaling factor to make the expression more comparable to $S_{\xx,\beta}$. Let $q := \lfloor \log S_{\xx,\beta} \rfloor$. Our scaling factor will be $e^{-2^{q-1}}$.
\begin{align*}
    S_{f(\xx),\beta}
    &
	= \frac{\sum_{\ell=2}^{D} 2^{\lfloor\log \ell - 1\rfloor} x_\ell e^{-2^{\lfloor\log \ell - 1\rfloor}\beta}}{\sum_{\ell=2}^{D} x_\ell e^{-2^{\lfloor\log \ell - 1\rfloor}\beta}}
    \ge \frac{\sum_{\ell=2}^{D} \frac l4 x_\ell e^{(-2^{q-1}-2^{\lfloor\log \ell - 1\rfloor})\beta}}{\sum_{\ell=2}^{D} x_\ell e^{(-2^{q-1}-2^{\lfloor\log \ell - 1\rfloor})\beta}}
    = \frac{\sum_{i=0}^{D}iw'_i}{4\sum_{i=0}^{D}w'_i}
    \enspace,
\end{align*}
where $w'_i = \begin{cases} x_i e^{(-2^{q-1}-2^{\lfloor \log i - 1\rfloor})\beta}&\text{if } i \ge 2,\\0  &\text{otherwise.}\end{cases}$

We set $p = 3 \cdot 2^{q-2}$, and verify that we meet all of the conditions of the Claim \ref{cl:avg}:

Firstly we need that all $w_i$ and $w'_i$ are non-negative, which is obviously the case.

Secondly we need that $p < \frac{\sum_{i=0}^D i w_i}{\sum_{i=0}^D w_i}$, which is true since $p < 2^q \le S_{\xx,\beta} = \frac{\sum_{i=0}^D i w_i}{\sum_{i=0}^D w_i}$.

Thirdly we need $w'_i\le w_i$ for all $i < p$ and $w'_i\ge w_i$ for all $i \ge p$. To show this, note that
\begin{displaymath}
    w'_i \ge w_i
        \iff
    (-2^{q-1}-2^{\lfloor\log i - 1\rfloor}) \beta \ge -i \beta
        \iff
    2^{q-1}+2^{\lfloor\log i - 1\rfloor} \le i
        \enspace.
\end{displaymath}

When $i \le 2^{q-1}$, clearly $2^{q-1}+2^{\lfloor\log i - 1\rfloor}> i$, so $w'_i \le w_i$.

When $2^{q-1} < i < p$, $2^{q-1}+2^{\lfloor\log i - 1\rfloor} = 2^{q-1}+2^{q-2} = p > i$, so $w'_i \le w_i$.

When $p \le i < 2^q$, $2^{q-1}+2^{\lfloor\log i - 1\rfloor} = 2^{q-1}+2^{q-2} = p \le i$, so $w'_i \ge w_i$.

When $2^q \le i$, $2^{q-1}+2^{\lfloor\log i - 1\rfloor} \le 2^{q-1}+2^{\log i - 1} \le 2^{q-1} + \frac i2 \le i$, so $w'_i \ge w_i$.

Therefore we have all the necessary conditions to apply Claim \ref{cl:avg}, yielding $\frac{\sum_{i=0}^D i w'_i}{\sum_{i=0}^D w'_i} > p$. Then,
\begin{displaymath}
    S_{f(\xx),\beta}
        \ge
    \frac{\sum_{i=0}^D i w'_i}{4 \sum_{i=0}^D w'_i}
        >
    \frac{p}{4}
        \ge
    \frac{3q}{16}
        >
    \frac{3S_{\xx,\beta}}{32}
        >
    \frac{S_{\xx,\beta}}{11}
    \enspace.
\end{displaymath}

This completes the proof of Claim \ref{cl:avg}.
\end{proof}


\subsubsection{Bounding $S_{\xx,\beta}$ in terms of $S_{g(\xx),\beta}$}
\label{subsubsec:clustering-property-2nd-simplification-of-x}

Having applied $f$ to ensure that only power-of-2 coefficients of $\xx$ are non-zero (cf. Claim \ref{cl:avg}), we apply a second transformation to ensure that every such coefficient is at least half the previous one.

Let $g: \mathbb{R}^{D+1} \rightarrow \mathbb{R}^{D+1}$ be given by $g(\xx)_i = \begin{cases}\sum_{\ell \le i} \frac {\ell x_\ell}{i}  &\text{if } i = 2^k\text{ for some }k\in \mathbb{N}_0,\\ x_i  &\text{otherwise.}\end{cases}$

This definition achieves our aim since when $i$ is a power of $2$,
\begin{displaymath}
    2g(\xx)_{2i}
        =
    2 \sum_{\ell \le 2i} \frac{\ell x_\ell}{2i}
        =
    \sum_{\ell \le 2i} \frac{\ell x_\ell}{i}
        \ge
    \sum_{\ell \le i} \frac{\ell x_\ell}{i}
        =
    g(\xx)_{2i}
        \enspace.
\end{displaymath}

\begin{claim}
\label{cl:trans2}
For all $\xx \in \mathbb{N}_0^D$ which have $x_i = 0$ for all $i \notin \{2^k:k\in \mathbb{N}_0\} $, $S_{\xx,\beta} \le 2S_{g(\xx),\beta}$.
\end{claim}

\begin{proof}
\begin{align*}
    S_{g(\xx),\beta}
        &
        =
    \frac{\sum_{i=0}^D i g(\xx)_i e^{-i\beta}}{\sum_{i=0}^D g(x)_i e^{-i\beta}}
        \ge
    \frac{\sum_{k=0}^{\log D} 2^k g(\xx)_{2^k} e^{-2^k\beta}}{\sum_{k=0}^{\log D} g(\xx)_{2^k} e^{-2^k\beta}}
        \ge
    \frac{\sum_{k=0}^{\log D} 2^k x_{2^k} e^{-2^k\beta}}{\sum_{k=0}^{\log D} \sum_{\ell=0}^k \frac {2^\ell x_{2^\ell}}{2^k} e^{-2^k\beta}}
        \\
        &
        =
    \frac{\sum_{k=0}^{\log D} 2^k x_{2^k} e^{-2^k\beta}}{\sum_{\ell=0}^{\log D} \sum_{k=\ell}^{\log D} \frac {2^\ell x_{2^\ell}}{2^k} e^{-2^\ell \beta}}
        \ge
    \frac{\sum_{k=0}^{\log D} 2^k  x_{2^k} e^{-2^k\beta}}{2\sum_{\ell=0}^{\log D} x_{2^\ell} e^{-2^\ell \beta}}
        \ge
    \frac{S_{\xx,\beta}}{2}
    \enspace.
\qedhere
\end{align*}
\end{proof}


\subsection{Bounding $S_{\xx,\beta}$ for simplifed $\xx$}
\label{subsec:clustering-property-simplification-of-x-usage}

Now that we have shown in Claims \ref{cl:trans1} and \ref{cl:trans2} that the transformations $f$ and $g$ do not increase $S_{\xx,\beta}$ by more than a constant factor, we show how they help to bound the value of $S_{\xx,\beta}$. Let $\xx'$ be the vector obtained after applying the two transformations to $\xx$, i.e., $\xx' = g \circ f(\xx)$.

\begin{claim}\label{clm:xprop}
$\xx'$ has the following properties:
\begin{itemize}
\item $x'_i = 0$ for all $i \notin \{2^k:k\in \mathbb N_0\}$;
\item $x'_1 \ge 2$;
\item $||\xx'||_1 \le 2n$, ;
\item $2x'_{2i} \ge x'_i$ for all $i$, due to transformation $g$.
\end{itemize}
\end{claim}

\begin{proof}
The first property is obvious due to transformation $f$. The second is true since $x'_1 \ge f(x)_1 = x_2+x_3 \ge 2$. The third is the case since $f$ does not increase $L_1$-norm and $g$ at most doubles it, and the fourth follows from transformation $g$.
\end{proof}

Our argument will be based on examining the ratios between consecutive non-zero coefficients in $\xx'$. To that end, define $k_i = \log \frac{x'_{2^{i+1}}} {x'_{2^i}}$ for all $0.01 \log D \le i \le 0.1 \log D$, and note that $k_i \ge \log \frac 12 = -1$ for all $i$ and $\sum_{i=0}^{\log D} k_i \le \log n$ due to the above properties.

We first show a condition on these $k_i$ which guarantees that $S_{\xx',\beta}$ (and therefore $S_{\xx,\beta}$) is $O(\frac{\log n}{\beta\log D})$ for some particular value of $\beta$:

\begin{claim}
\label{cl:goodj}
If for fixed $j$ and for all $m \ge 8$ we have $\sum\limits_{\ell = j + \log{\frac{\log n}{\log D}}}^{j + \log{\frac{\log n}{\log D}} + m} k_\ell \le 2^m \frac{\log n}{\log D}$, then
$S_{\xx',2^{-j}} = O(\frac{2^j \log n}{\log D})$.
\end{claim}

\begin{proof}
We first split $T_{\xx',2^{-j}}$ (the numerator of $S_{\xx',2^{-j}}$) into three parts, which we will bound separately (recall that $B_{\xx',2^{-j}} = \sum_{i=0}^D x'_{i} e^{-i2^{-j}}$):

\begin{align*}
    T_{\xx',2^{-j}}
        =
    \sum_{i=0}^D i x'_i e^{-i2^{-j}}
        =
    \sum_{i=0}^{\log D} 2^i x'_{2^i} e^{-2^{i-j}}
    = P+Q+R
        \enspace,
\end{align*}

where $P = \sum\limits_{i=0}^{j+\log{\frac{\log n}{\log D}}+8} 2^i x'_{2^i} e^{-2^{i-j}}$, $Q = \sum\limits_{i=j+\log{\frac{\log n}{\log D}}+9}^{j+\log\log n} 2^i x'_{2^i} e^{-2^{i-j}}$, and $R = \sum\limits_{i=\log\log n+1}^{\log D} 2^i x'_{2^i} e^{-2^{i-j}}$.

We now bound these parts. $P$ is the largest, and we require that $P = O(\frac{2^j\log n}{\log D})B_{\xx',2^{-j}}$ .

\begin{align*}
    P
        &
        =
    \sum_{i=0}^{j+\log{\frac{\log n}{\log D}}+8} 2^i x'_{2^i} e^{-2^{i-j}}
        \le
    \sum_{i=0}^{j+\log{\frac{\log n}{\log D}}+8} 256\frac{2^j\log n}{\log D} x'_{2^i} e^{-2^{i-j}}
        \\
        &
        \le
    256\frac{2^j\log n}{\log D}\sum_{i=0}^{\log D}  x'_{2^i} e^{-2^{i-j}}
        =
    256\frac{2^j\log n}{\log D}B_{\xx',2^{-j}}
        \enspace.
\end{align*}

Using the condition of Claim \ref{cl:goodj}, we can show that $Q$ is also $O(\frac{2^j\log n}{\log D})B_{\xx',2^{-j}}$. Let $m \ge 9$. We begin by re-expressing $x'_{\frac{2^{j+m}\log n}{\log D}}$:
\begin{displaymath}
    x'_{\frac{2^{j+m}\log n}{\log D}}
        =
    x'_{\frac{2^{j}\log n}{\log D}}
    \prod_{\ell=j+\log{\frac{\log n}{\log D}}}^{j+\log{\frac{\log n}{\log D}}+m-1} 
        \frac{x'_{2^{\ell+1}}}{x'_{2^\ell}}
        =
    x'_{\frac{2^{j}\log n}{\log D}}
        2^{\sum\limits_{\ell=j+\log{\frac{\log n}{\log D}}}^{j+\log{\frac{\log n}{\log D}}+m-1}k_\ell}
        \enspace.
\end{displaymath}

We can then apply the condition of the claim:
\begin{align*}
    x'_{\frac{2^{j+m}\log n}{\log D}}
        &\le
    x'_{\frac{2^{j}\log n}{\log D}} 2^{2^{m-1}\frac{\log n}{\log D}}
        \le
    e^{\frac{2^{j}\log n}{\log D}2^{-j}} 2^{2^{m-1}\frac{\log n}{\log D}} x'_{\frac{2^{j}\log n}{\log D}} e^{-\frac{2^{j}\log n}{\log D}2^{-j}}
        \\
        &\le
    e^{\frac{\log n}{\log D}} 2^{2^{m-1}\frac{\log n}{\log D}}\sum_{i=0}^{D}  x'_{i} e^{-i2^{-j}}
        =
    2^{(2^{m-1}+\log e)\frac{\log n}{\log D}}B_{\xx',2^{-j}}
        \enspace.
\end{align*}

We can use this to bound $Q$ as follows,
\begin{align*}
    Q
        &=
    \sum_{i=j+\log{\frac{\log n}{\log D}}+9}^{j+\log\log n} 2^i x'_{2^i} e^{-2^{i-j}}
        =
    \frac{2^{j}\log n}{\log D}\sum_{m=9}^{\log\log n} 2^m x'_{\frac{2^{j+m}\log n}{\log D}} e^{-2^{m+\log{\frac{\log n}{\log D}}}}
        \\
        &\le
	\frac{2^j\log n}{\log D}\sum_{m=9}^{\log\log n} 2^m \cdot 2^{(2^{m-1}+\log e)\frac{\log n}{\log D}}B_{\xx',2^{-j}} \cdot e^{-2^{m+\log{\frac{\log n}{\log D}}}}
        \enspace.
\end{align*}

Rearranging terms, we obtain,
\begin{align*}
	Q
        &=
	\frac{2^{j}\log n}{\log D}B_{\xx',2^{-j}}\sum_{m=9}^{\log\log n} 2^{m+(2^{m-1}+\log e)\frac{\log n}{\log D}-2^{m}\frac{\log n}{\log D}}
        \\
        &\le
	\frac{2^{j}\log n}{\log D}B_{\xx',2^{-j}}\sum_{m=9}^{\log\log n} 2^{-2^{m-2}\frac{\log n}{\log D}}
        \le
	\frac{2^{j}\log n}{\log D}B_{\xx',2^{-j}}
\enspace.
\end{align*}

$R$ is always negligible, since the $e^{-2^{i-j}}$ term is very small for large $i$.
\begin{align*}
    R
        =
    \sum_{i=j+\log\log n+1}^{\log D} 2^i x'_{2^i} e^{-2^{i-j}}
        \le
    \sum_{i=j+\log\log n+1}^{\log D} D x'_{2^i} e^{-2\log n}
        \le
    2Dn^{1-2\log e}
        \le
    1
    \enspace.
\end{align*}

So,
\begin{align*}
    S_{\xx',2^{-j}}
        &=
    \frac{P+Q+R}{B_{\xx',2^{-j}} }
        \le
    \frac{256\frac{2^j\log n}{\log D}B_{\xx',2^{-j}} +\frac{2^{j}\log n}{\log D}B_{\xx',2^{-j}} +1}{B_{\xx',2^{-j}}}
        \le
    258\frac{2^j\log n}{\log D}
    \enspace.
    \qedhere
\end{align*}
\end{proof}

Finally, we must show that there are many $j$ for which the condition of Claim \ref{cl:goodj} holds.

\begin{claim}
\label{cl:goodj-cond}
The number of integers $j$, $0.01\log D \le j \le 0.1\log D$, for which there is $i \ge 8$ satisfying $\sum_{\ell=j+\log{\frac{\log n}{\log D}}}^{j+\log{\frac{\log n}{\log D}}+i}k_\ell > 2^{i}\frac{\log n}{\log D}$ is upper bounded by $0.04\log D$.
\end{claim}

\begin{proof}
Consider the following process: take values of $i\in [0.01\log D,0.1\log D]$ in increasing order. If there is some $i'\ge i+8$ such that $\sum_{\ell=i}^{i'}k_\ell > 2^{i'-i}\frac{\log n}{\log D}$, then call all values between $i$ and the largest such $i'$ `bad', and continue the process from $i'+1$. Let $b$ denote the number of bad $i$. The average $k_i$ over all bad $i$ must be at least $\frac{2^8\log n}{9 \log D}$, and since all $k_i$ are bounded below by $-1$ and sum to at most $\log n$, we have
\begin{displaymath}
    \frac{2^8\log n}{9 \log D}b+(-1)(0.09\log D - b) \le \log n
        \enspace,
\end{displaymath}
and so $b \le \frac{\log n+0.09\log D}{\frac{2^8\log n}{9 \log D}+1}\le \frac{1.09\log n}{\frac{2^8\log n}{9 \log D}}\le 0.04\log D$. For every $j$ in the set, $j+\log{\frac{\log n}{\log D}}$ must be bad, and so the size of the set is also at most $0.04\log D$.
\end{proof}

We are now ready to prove our main result, Theorem \ref{thm:cprop}.

\begin{proof}[Proof of Theorem \ref{thm:cprop}]
With probability at least $1-\frac{0.04}{0.1-0.01} \ge 0.55$, for all $i\ge 8$ we have that
\begin{displaymath}
    \sum_{\ell=j+\log{\frac{\log n}{\log D}}}^{j+\log{\frac{\log n}{\log D}}+i} k_\ell
        \le 
    2^{i}\frac{\log n}{\log D}
        \enspace.
\end{displaymath}
Then, $S_{\xx',2^{-j}} = O(\frac{2^j\log n}{\log D})$ by Claims \ref{cl:goodj} and \ref{cl:goodj-cond}. Applying Claims \ref{cl:trans1} and \ref{cl:trans2}, we get $S_{\xx,2^{-j}} = O(\frac{2^j\log n}{\log D})$. Finally, applying Lemma \ref{lem:edist}, we find that the expected distance from $v$ to its cluster center is at most $O(\frac{2^j\log n}{\log D})$.
\end{proof}

\end{document}